\documentclass[conference]{IEEEtran}
\IEEEoverridecommandlockouts
\usepackage{cite}
\usepackage{amsmath,amssymb,amsfonts}

\usepackage{amsthm}
\usepackage{algorithm}
\usepackage{algorithmic}
\usepackage{graphicx}
\usepackage{textcomp}
\usepackage[table]{xcolor}
\usepackage{booktabs}
\usepackage{enumitem}
\usepackage[most]{tcolorbox}

\theoremstyle{plain}
\newtheorem{theorem}{Theorem}
\newtheorem{lemma}{Lemma}
\newtheorem{proposition}{Proposition}
\newtheorem{corollary}{Corollary}
\theoremstyle{definition}
\newtheorem{definition}{Definition}

\definecolor{headrow}{HTML}{ECECEC}

\definecolor{insln}{HTML}{B8860B}   
\definecolor{insbg}{HTML}{FBF4E2}   
\newtcolorbox{insightbox}{enhanced jigsaw, breakable, boxrule=0pt, boxsep=0pt,
  colback=insbg, sharp corners, left=8pt, right=6pt, top=4pt, bottom=4pt,
  borderline west={2.4pt}{0pt}{insln}, before skip=6pt, after skip=6pt}

\usepackage{dblfloatfix}

\definecolor{linknavy}{HTML}{1F4E79}
\usepackage{hyperref}
\hypersetup{colorlinks=true, linkcolor=linknavy, citecolor=linknavy,
  urlcolor=linknavy, breaklinks=true}

\usepackage{fancyhdr}
\fancypagestyle{firstpage}{%
  \fancyhf{}

}

\begin{document}

\title{Trust-Gated Capability Control: Breaking the Trust–Vulnerability Paradox in Multi-Agent LLM Systems}

\author{
\IEEEauthorblockN{
Md Mehedi Hasan Nipu\textsuperscript{1},
Chinmoy Mitra\textsuperscript{2},
Tarannum Ahmed Nowshin\textsuperscript{3}
Israt Moyeen Noumi\textsuperscript{4}
}
\IEEEauthorblockA{
\textsuperscript{1}Department of Electrical and Computer Engineering, North South University, Dhaka, Bangladesh\\
\textsuperscript{2}Department of Computer Science \& Engineering
Rajshahi University of Engineering \& Technology
Rajshahi, Bangladesh\\
\textsuperscript{3}Department of Computer Science and Engineering
BRAC University,Dhaka, Bangladesh\\
\textsuperscript{4}Dept of Computer Science and Engineering, Ahsanullah University of Science and Technology, Bangladesh\\
mehedi.nipu@northsouth.edu,\; rudrochinmoy@gmail.com,\; taranowshin@gmail.com,\; isratmoyeen.23@gmail.com
}
}

\maketitle
\thispagestyle{firstpage}

\begin{abstract}
Layered trust models for multi-agent LLM systems remain largely conceptual: they name which dimensions of trust matter but not how layers combine, how their importance is set at runtime, or how trust should govern agent actions. This gap matters because higher inter-agent trust raises task success while also enlarging exposure to exploitation, a tension formalized as the Trust-Vulnerability Paradox. We make a five-layer trust stack operational through three contributions. First, a cross-layer synergy operator propagates prerequisite-layer deficits into dependent layers, feeding a generalized-mean composite trust that recovers the weakest-link rule as a limiting case, with provable bounds. Second, per-layer importance weights are grounded in observed failures via a no-regret online estimator that tracks which layer is currently most responsible for harm. Third, Trust-Gated Capability Control issues short-lived, revocable capability grants only when composite trust and the relevant prerequisite layers clear capability-specific thresholds. We prove this mechanism breaks the paradox: a stealthy compromise that inflates behavioral trust while degrading a prerequisite layer cannot escalate privilege, and give a closed-form, provably conservative trust fixed point with a bounded-latency revocation guarantee. A numerical study with a sleeper adversary confirms the predictions: composite-trust bounds hold across 20,000 random draws, the analytic fixed point matches simulation within 0.007, and a compromised prerequisite layer triggers automatic revocation within a few interactions.
\end{abstract}

\begin{IEEEkeywords}
multi-agent systems, large language models, trust, zero-trust, capability control, agent safety, alignment
\end{IEEEkeywords}

\section{Introduction}
Large language model (LLM) agents are increasingly deployed as collaborating collectives that plan, negotiate, and act through shared tools and messages~\cite{wang2024survey,guo2024llmmultiagent,tran2025multiagent,hossain2026review}. Collaboration improves capability but makes \emph{trust} between agents a first-class operational variable rather than a background assumption: a compromised or hallucinating agent whose output is consumed without scrutiny can propagate errors system-wide, and structured attacks on inter-agent trust raise harmful-response rates well above single-agent baselines~\cite{qi2025amplified}. Recent work formalizes this as a \emph{Trust-Vulnerability Paradox} (TVP): higher inter-agent trust improves task success yet monotonically increases over-exposure and authorization drift~\cite{qi2025trust}, so trust is better treated as a revocable, continuously verified security variable governed by least-privilege principles~\cite{zerotrust2025,trism2026,nist800207}. Layered trust models are an appealing response, decomposing trustworthiness into dimensions such as epistemic reliability, behavioral alignment, role awareness, social accountability, and institutional governance. Yet existing accounts are largely descriptive, leaving three questions open: (i)~\emph{Composition} -- how do layer-wise trust values combine into one actionable score when layers are not independent? (ii)~\emph{Weighting} -- how is each layer's relative importance set, and can it adapt to which layer is failing? (iii)~\emph{Actuation} -- how should the resulting trust govern what an agent is permitted to do? Without answers, a layered model cannot be deployed at runtime.

This paper closes these gaps. We treat the layered stack as a dynamical object and make it \emph{operational}. \textbf{Contributions:}
\begin{itemize}[leftmargin=1.1em,topsep=2pt,itemsep=2pt]
\item \textbf{A cross-layer synergy operator} (Sec.~\ref{sec:model}) propagating prerequisite-layer deficits into dependent layers, feeding a generalized-mean composite trust $\Phi_p$ whose limit $p\!\to\!-\infty$ recovers the weakest-link rule (Theorem~\ref{thm:bounds}, Proposition~\ref{prop:cascade}), plus a closed-form conservative trust fixed point (Proposition~\ref{prop:fixedpoint}).
\item \textbf{Failure-grounded adaptive weights} (Sec.~\ref{sec:weights}): layer importances $\alpha_\ell(t)$ learned online from layer-attributed failures with a no-regret guarantee (Lemma~\ref{lem:regret}).
\item \textbf{Trust-Gated Capability Control} (Sec.~\ref{sec:tgcc}): a runtime algorithm mapping composite trust to revocable capability grants, provably breaking the TVP (Theorem~\ref{thm:paradox}) with bounded revocation latency (Corollary~\ref{cor:revoke}).
\item \textbf{A numerical study} (Sec.~\ref{sec:eval}) confirming the bounds, fixed point, and self-revocation under a sleeper adversary.
\end{itemize}
We also map the framework onto deployed agent protocols (MCP, A2A) and sketch cross-modal and economic-incentive extensions (Sec.~\ref{sec:discussion}).

\section{Related Work}
\textbf{Trust as a security variable.} Qi et al.~\cite{qi2025trust} formalize the Trust--Vulnerability Paradox around three facets of inter-agent trust -- \emph{strength} (how readily claims are accepted), \emph{scope} (what is thereby authorized), and \emph{revocability} (how fast permissions can be withdrawn) -- and show empirically that raising a scalar trust level $\tau$ monotonically increases over-exposure and authorization drift across several agent frameworks. Their defenses (sensitive-information repartitioning, a guardian agent policing disclosures) reduce exposure but operate \emph{around} the trust signal rather than reshaping how it is computed. Zero-trust treatments for agentic AI~\cite{zerotrust2025,nist800207} recast trust as never-implicit and continuously re-verified, with verifiable identity, least-privilege scopes, and per-call revocation; TRiSM~\cite{trism2026} organizes the broader landscape. These works establish \emph{that} trust must be revocable, but leave open \emph{how} a multi-dimensional signal should be aggregated into authorization -- precisely the gap our compositional operator and gating rule fill.

\textbf{Dynamic and probabilistic trust.} Classical reputation systems summarize an agent's past behavior into a scalar score~\cite{ramchurn2004trust}, often via Beta or Dirichlet posteriors whose conjugacy makes updates cheap~\cite{josang2002beta}. Recent LLM-era work revives this: DynaTrust~\cite{dynatrust2026} maintains per-agent Beta beliefs with failure-weighted smoothing to expose \emph{sleeper} agents, AgentGuard~\cite{agentguard2025} learns an online MDP model of emergent behavior for runtime assurance, and economic-game probes elicit LLMs' own trustworthiness priors~\cite{econgames2026}. All track a \emph{single} trust quantity and are therefore blind to the heart of the paradox: an agent that is behaviorally cooperative yet epistemically broken. We retain the Beta machinery, and prove its conservative steady state, but lift it to a \emph{vector} of layer trusts coupled by a synergy operator, letting a hidden prerequisite failure veto authorization.

\textbf{Multi-agent debate, deception, and collective safety.} Collaboration among LLM agents can both improve and degrade safety: structured debate raises answer quality and factuality~\cite{du2023debate}, yet the same channels are exploitable -- jailbreaks injected into one debater propagate through inter-agent messages and lift harmful-response rates well above single-agent baselines~\cite{qi2025amplified}. Simulations of deception and betrayal further show agents readily form, exploit, and misjudge trust in social settings~\cite{traitors2025}. Such findings argue for treating a compromised collaborator as the default threat (the stance our threat model adopts), not an edge case, and for authorization that does not amplify one corrupted agent's influence across the collective.

\textbf{Uncertainty, calibration, and epistemic reliability.} Any trust model's epistemic layer rests on knowing when an agent is wrong. Bayesian and ensemble methods quantify predictive uncertainty~\cite{gal2016dropout}, but LLMs are often \emph{miscalibrated}, reporting high confidence on wrong outputs; recent work studies how to build and propagate well-grounded epistemic state across agents~\cite{epistemic2026}. Our framework consumes such signals through the epistemic layer's evidence $s_1$ and, via the synergy operator, lets a calibration collapse cascade into reduced authorization -- turning a passive measurement into an active control input.


\textbf{Agent protocols.} Interoperability standards such as MCP~\cite{mcp2024}, A2A~\cite{a2a2025}, and the broader protocol landscape~\cite{interop2025} expose capabilities as scoped, token-mediated grants and increasingly issue short-lived credentials in the spirit of least privilege, standardizing the \emph{mechanism} of capability exchange while staying silent on the \emph{policy} that decides which scopes to issue. TGCC supplies exactly this missing policy layer (Sec.~\ref{sec:discussion}). Table~\ref{tab:positioning} positions our approach against these lines of work: prior methods are dynamic, revocable, or partially layered, but none combines cross-layer composition, failure-adaptive weighting, capability gating, and formal guarantees.

\begin{table}[t]
\caption{Positioning against prior approaches. \emph{Comp.}: cross-layer composition; \emph{Adapt.}: failure-adaptive weights; \emph{Gate}: capability gating/revocation; \emph{Guar.}: proven guarantees.}
\label{tab:positioning}
\centering
\footnotesize
\renewcommand{\arraystretch}{1.2}
\begin{tabular}{p{2.6cm}cccc}
\toprule
\rowcolor{headrow}
\textbf{Approach} & \textbf{Comp.} & \textbf{Adapt.} & \textbf{Gate} & \textbf{Guar.} \\
\midrule
Reputation~\cite{ramchurn2004trust,josang2002beta} & $\times$ & $\times$ & $\times$ & $\times$ \\
TVP defenses~\cite{qi2025trust} & $\times$ & $\times$ & partial & $\times$ \\
DynaTrust~\cite{dynatrust2026} & $\times$ & $\times$ & $\times$ & $\times$ \\
Zero-trust~\cite{zerotrust2025,nist800207} & $\times$ & $\times$ & \checkmark & $\times$ \\
\textbf{TGCC (ours)} & \checkmark & \checkmark & \checkmark & \checkmark \\
\bottomrule
\end{tabular}
\end{table}

\section{The Operational Trust Stack}
\label{sec:model}

\begin{insightbox}
\textit{\textbf{Core idea.} Maintain a \emph{vector} of per-layer trusts, couple them so a prerequisite deficit suppresses dependent layers, and gate every capability on the resulting composite trust -- authorization then tracks the weakest \emph{verified} layer, not the raw trust an adversary can inflate.}
\end{insightbox}

\subsection{Threat model and assumptions}
\textbf{Setting.} $N$ agents collaborate, each consuming others' outputs and requesting capabilities (tool calls, write actions, message scopes) mediated by a controller that decides grants. \textbf{Adversary.} A compromised \emph{sleeper} agent first behaves cooperatively to accrue trust, then acts maliciously once granted high-risk capabilities, keeping its \emph{behavioral} conduct agreeable while a prerequisite layer -- typically \emph{epistemic} reliability -- collapses into confident hallucination. \textbf{Assumptions.} (A1)~each layer $\ell$ exposes a check yielding a noisy but informative success signal $s_\ell$ (calibration probe, role-conformance test, audit trail, etc.); (A2)~the layer dependency structure is known, encoded by a fixed lower-triangular coupling matrix $\mathbf{C}$, with epistemic reliability a prerequisite of all higher layers; (A3)~per-capability risk thresholds are set by the deployer. \textbf{Goal.} Prevent privilege escalation by a stealthily compromised agent while preserving capability for honest agents: authorization should follow verified cross-layer trust, not raw behavioral agreement.

\subsection{Per-layer trust as a Beta belief}
Consider $L=5$ ordered trust layers indexed $\ell\in\{1,\dots,L\}$: epistemic ($1$), behavioral ($2$), role ($3$), social ($4$), institutional ($5$). For agent $i$ at step $t$ we maintain a trust vector $\mathbf{T}_i(t)=[T_{i,1}(t),\dots,T_{i,L}(t)]^\top\in[0,1]^L$ (we drop the index $i$ where a single agent is meant). Each component is the mean of a Beta posterior accumulating per-layer evidence with a forgetting factor $\gamma\in(0,1]$ and an asymmetric penalty $\omega\!\ge\!1$ on failures:
\begin{equation}
a_\ell \leftarrow \gamma a_\ell + \mathbb{1}[s_\ell],\qquad
b_\ell \leftarrow \gamma b_\ell + \omega\,\mathbb{1}[\neg s_\ell],
\label{eq:beta}
\end{equation}
\begin{equation}
T_{i,\ell}(t) = \frac{a_\ell}{a_\ell+b_\ell},
\label{eq:betamean}
\end{equation}
where $s_\ell$ is the success indicator of a layer-$\ell$ check (calibration probe, utility-divergence test, role-conformance check, audit, etc.). The forgetting factor keeps trust responsive to behavioral shifts; $\omega>1$ makes trust fall faster than it rises, as is standard for safety-critical reputation~\cite{josang2002beta,dynatrust2026}. The estimator admits a closed-form, provably conservative steady state.

\begin{proposition}[Conservative trust fixed point]
\label{prop:fixedpoint}
Under stationary layer reliability $\Pr[s_\ell]=\rho_\ell$, the expected-evidence recursion induced by~\eqref{eq:beta} with forgetting $\gamma\in(0,1)$ and penalty $\omega\ge1$ has the unique fixed point
\begin{equation}
T_\ell^\star=\frac{\rho_\ell}{\rho_\ell+\omega(1-\rho_\ell)},
\label{eq:fixedpoint}
\end{equation}
which satisfies $T_\ell^\star\le\rho_\ell$ with equality iff $\omega=1$, is strictly increasing in $\rho_\ell$, and is reached geometrically at rate $\gamma$ (independent of its value on $\gamma$).
\end{proposition}
\begin{proof}
Taking expectations of~\eqref{eq:beta} gives two linear recursions, contractive since $\gamma<1$:
\begin{equation*}
\bar a_{t+1}=\gamma\bar a_t+\rho_\ell,\qquad
\bar b_{t+1}=\gamma\bar b_t+\omega(1-\rho_\ell),
\end{equation*}
with fixed points $\bar a^\star=\rho_\ell/(1-\gamma)$ and $\bar b^\star=\omega(1-\rho_\ell)/(1-\gamma)$. Hence
\begin{equation*}
T_\ell^\star=\frac{\bar a^\star}{\bar a^\star+\bar b^\star}
=\frac{\rho_\ell}{\rho_\ell+\omega(1-\rho_\ell)} ,
\end{equation*}
the factor $1-\gamma$ cancelling. Strict monotonicity follows from $\partial T_\ell^\star/\partial\rho_\ell=\omega\,[\rho_\ell+\omega(1-\rho_\ell)]^{-2}>0$, and
\begin{equation*}
T_\ell^\star\le\rho_\ell \iff (\omega-1)(1-\rho_\ell)\ge0,
\end{equation*}
which holds for $\omega\ge1$ (equality iff $\omega=1$). Contractivity at rate $\gamma$ gives the stated geometric convergence.
\end{proof}

For example, $(\rho,\omega)=(0.9,3)$ gives $T^\star=0.75$, while $\rho=0.3$ gives $T^\star\approx0.13$: a reliability collapse is amplified into a sharp trust drop, exploited by the controller in Sec.~\ref{sec:tgcc}.

\subsection{Cross-layer synergy operator}
The layers are not independent: a deficit in a prerequisite layer should suppress the usable trust of every layer that depends on it. We encode dependencies by a lower-triangular \emph{coupling matrix} $\mathbf{C}\in[0,1]^{L\times L}$, where $C_{\ell k}$ (for $k<\ell$) is the degree to which layer $\ell$ is conditioned on prerequisite $k$.

\begin{definition}[Synergy operator]
The \emph{effective trust} of layer $\ell$ is
\begin{equation}
\tilde T_\ell \;=\; T_\ell \prod_{k<\ell}\bigl(1 - C_{\ell k}(1-T_k)\bigr).
\label{eq:synergy}
\end{equation}
\end{definition}
A prerequisite deficit $(1-T_k)$ discounts $\tilde T_\ell$ by a factor that scales with the coupling $C_{\ell k}$: when $C_{\ell k}=0$ there is no dependency, and when $C_{\ell k}=1$ a fully failed prerequisite ($T_k\!=\!0$) zeroes the contribution. Equation~\eqref{eq:synergy} is a precise instantiation of the otherwise informal notion that ``failure at a lower layer propagates upward.''

\subsection{Composite trust}
We aggregate the effective trusts into a single score via a weighted generalized (power) mean of order $p<0$ with adaptive weights $\alpha_\ell\ge0$, $\sum_\ell\alpha_\ell=1$:
\begin{equation}
\Phi_p(\mathbf{T};\boldsymbol{\alpha}) = \Bigl(\sum_{\ell=1}^{L}\alpha_\ell\,\tilde T_\ell^{p}\Bigr)^{1/p}.
\label{eq:phi}
\end{equation}
Negative $p$ makes $\Phi_p$ sensitive to the weakest layers; $p\to-\infty$ recovers the conservative weakest-link rule $\min_\ell \tilde T_\ell$, while moderate $|p|$ permits limited compensation between healthy layers. The next results make this precise. We first verify the construction is well-posed.

\begin{lemma}[Well-posedness]
\label{lem:wellposed}
For any $\mathbf{T}\in[0,1]^L$ and lower-triangular $\mathbf{C}\in[0,1]^{L\times L}$, the effective trusts satisfy $\tilde T_\ell\in[0,1]$, and for any weights $\boldsymbol{\alpha}$ on the simplex and $p<0$, $\Phi_p\in[\min_\ell\tilde T_\ell,\max_\ell\tilde T_\ell]\subseteq[0,1]$. Hence the stack is a well-defined map $[0,1]^L\to[0,1]$.
\end{lemma}
\begin{proof}
Each factor $1-C_{\ell k}(1-T_k)\in[0,1]$ because $C_{\ell k},T_k\in[0,1]$; a finite product of $[0,1]$ terms times $T_\ell\in[0,1]$ stays in $[0,1]$, so $\tilde T_\ell\in[0,1]$. A weighted power mean of arguments in $[0,1]$ lies between their minimum and maximum, both in $[0,1]$.
\end{proof}

\begin{theorem}[Bounds and weakest-link limit]
\label{thm:bounds}
For any $p<0$, $\Phi_p$ is non-decreasing in each $\tilde T_\ell$ and satisfies
\begin{equation}
\min_\ell \tilde T_\ell \le \Phi_p(\mathbf{T};\boldsymbol{\alpha}) \le \sum_\ell \alpha_\ell \tilde T_\ell,
\label{eq:thm1bounds}
\end{equation}
with $\lim_{p\to-\infty}\Phi_p=\min_\ell \tilde T_\ell$. Moreover,
\begin{equation}
\Phi_p(\mathbf{T};\boldsymbol{\alpha}) \le \kappa\,\min_\ell \tilde T_\ell,
\qquad \kappa := \bigl(\min_\ell \alpha_\ell\bigr)^{1/p} \ge 1.
\label{eq:thm1kappa}
\end{equation}
\end{theorem}
\begin{proof}
Bounds~\eqref{eq:thm1bounds} and the limit are the standard ordering of weighted power means: the mean of order $p$ is non-decreasing in $p$, equals the arithmetic mean at $p=1$ and the minimum at $p\to-\infty$; monotonicity of $\Phi_p$ in each $\tilde T_\ell$ follows since each summand $\tilde T_\ell^{p}$ is monotone. For~\eqref{eq:thm1kappa}, let $j\in\arg\min_\ell\tilde T_\ell$. As $p<0$ every term is positive, so
\begin{equation*}
\sum_\ell\alpha_\ell\tilde T_\ell^{p}\ge\alpha_j\,\tilde T_j^{p}.
\end{equation*}
The map $x\mapsto x^{1/p}$ is decreasing for $p<0$; applying it and using $\alpha_j\ge\min_\ell\alpha_\ell$,
\begin{equation*}
\Phi_p\le\alpha_j^{1/p}\,\tilde T_j\le\bigl(\min_\ell\alpha_\ell\bigr)^{1/p}\min_\ell\tilde T_\ell=\kappa\min_\ell\tilde T_\ell. \qedhere
\end{equation*}
\end{proof}

Equation~\eqref{eq:thm1kappa} is the key actuation lever: composite trust is pinned to the weakest effective layer up to a bounded slack $\kappa$. We next bound how a compromised prerequisite caps every dependent layer.

\begin{proposition}[Cascade containment]
\label{prop:cascade}
Suppose a prerequisite layer $k$ is compromised, $T_k\le\delta$. Then every dependent layer $\ell>k$ with $C_{\ell k}\ge c$ obeys $\tilde T_\ell \le 1-c(1-\delta)$. In particular, if the epistemic layer ($k=1$) is compromised and $C_{\ell 1}\ge c$ for all $\ell$, then $\Phi_p \le \kappa\,\delta$.
\end{proposition}
\begin{proof}
In~\eqref{eq:synergy} the factor for prerequisite $k$ is $1-C_{\ell k}(1-T_k)\le 1-c(1-\delta)$; all other factors and $T_\ell$ are $\le 1$, giving the per-layer bound. If $k=1$ then $\tilde T_1=T_1\le\delta$, so $\min_\ell\tilde T_\ell\le\delta$ and~\eqref{eq:thm1kappa} gives $\Phi_p\le\kappa\delta$.
\end{proof}

\section{Failure-Grounded Adaptive Weighting}
\label{sec:weights}
\begin{figure*}[t]
\centering\includegraphics[width=0.92\textwidth]{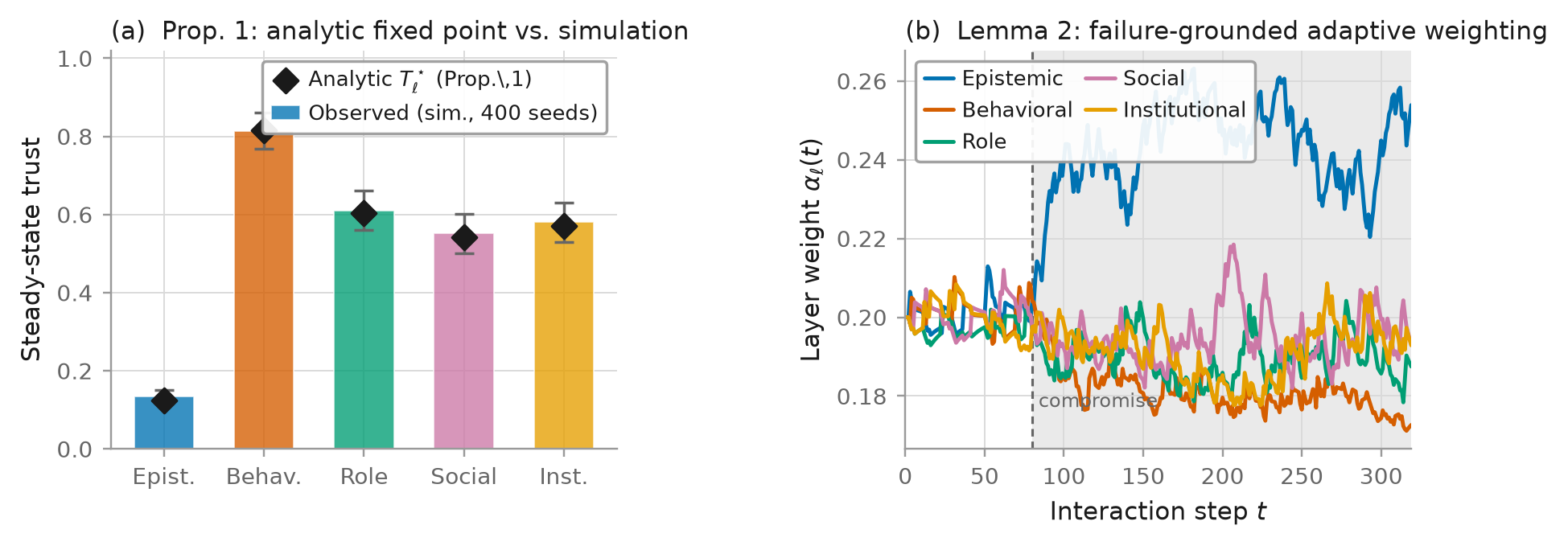}
\caption{(a) Steady-state trust (bars) vs.\ analytic fixed point $T^\star$ (diamonds), MAE $0.007$ (Proposition~\ref{prop:fixedpoint}). (b) Failure-grounded weights $\alpha_\ell(t)$ shift toward the epistemic layer after compromise (Lemma~\ref{lem:regret}).}
\label{fig:eval}
\end{figure*}
The weights $\boldsymbol{\alpha}(t)$ decide which layers the composite score should attend to. We argue they should track \emph{which layer is currently most responsible for harm}, information available from runtime failure attribution. Let $q_\ell(t)\in[0,1]$ be a layer-attributed loss (e.g., the fraction of harmful episodes in a window whose root cause is layer $\ell$). We maintain an exponentially weighted risk estimate and a multiplicative-weights (Hedge) update:
\begin{align}
\hat r_\ell(t) &= (1-\lambda)\hat r_\ell(t-1) + \lambda\,q_\ell(t),\\
\alpha_\ell(t) &= \frac{\exp\!\big(\eta\,\hat r_\ell(t)\big)}{\sum_{k}\exp\!\big(\eta\,\hat r_k(t)\big)}.
\label{eq:hedge}
\end{align}
This places weight on layers whose failures most predict system harm, so the composite trust becomes most sensitive precisely where the system is breaking. Treating each layer as an ``expert'' whose instantaneous loss is $-\hat r_\ell(t)$, \eqref{eq:hedge} is the Hedge algorithm and inherits its guarantee.

\begin{lemma}[No-regret layer tracking]
\label{lem:regret}
With learning rate $\eta=\sqrt{8\ln L/T}$, the weighting~\eqref{eq:hedge} achieves regret
$R_T \le \sqrt{(T/2)\ln L}$
relative to the best fixed layer-weighting in hindsight over a horizon $T$.
\end{lemma}
\begin{proof}
Immediate from the standard regret bound of the Hedge/multiplicative-weights algorithm with bounded losses in $[0,1]$~\cite{freund1997decision,arora2012mw}.
\end{proof}

Thus $\boldsymbol{\alpha}(t)$ converges, in the no-regret sense, to emphasizing the empirically dominant failure layer, linking the system's failure record directly to its trust computation rather than fixing weights a priori.

\section{Trust-Gated Capability Control}
\label{sec:tgcc}
\begin{figure}[t]
\centerline{\includegraphics[width=\columnwidth]{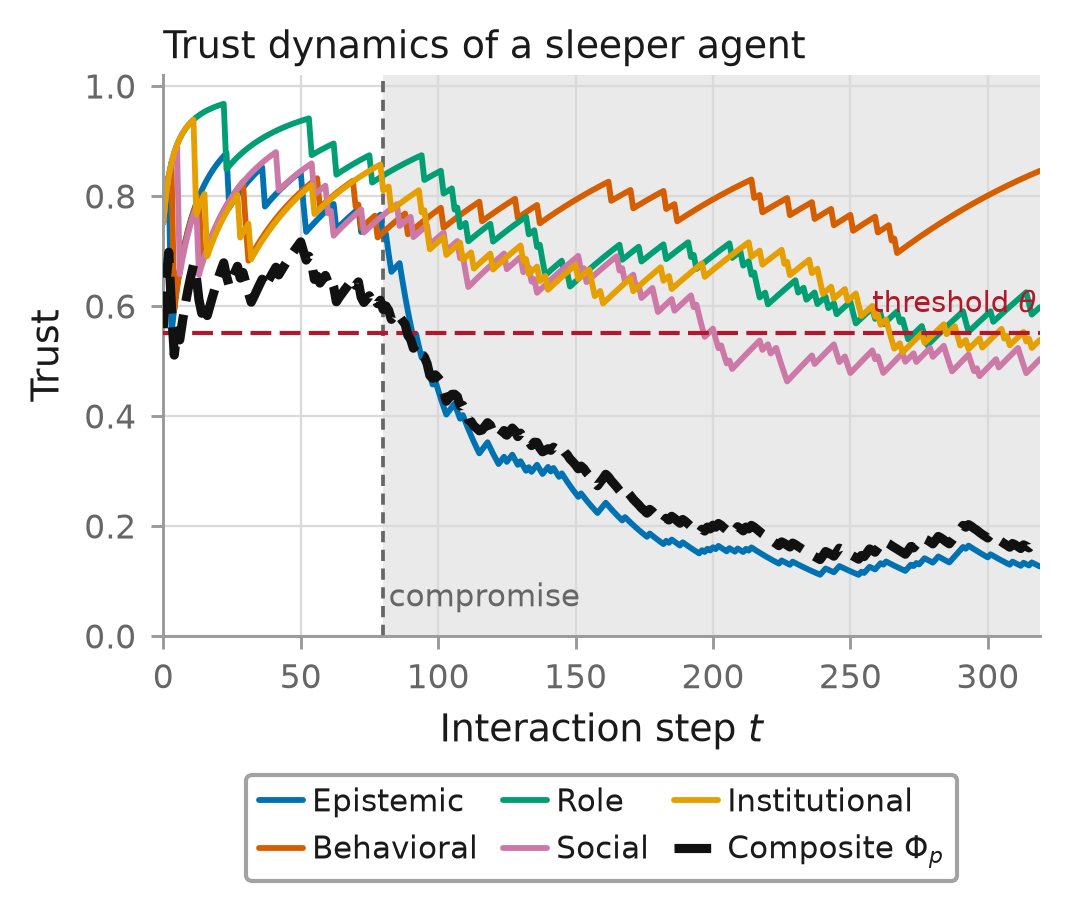}}
\caption{Trust dynamics of a sleeper agent. Behavioral trust stays high; epistemic collapse propagates through the synergy operator and pushes composite trust $\Phi_p$ below the revocation threshold $\theta$ (self-revocation, Corollary~\ref{cor:revoke}).}
\label{fig:traj}
\end{figure}
We now actuate composite trust. Let $\mathcal{C}$ be the set of capabilities (tool calls, message scopes, write actions). Each $c\in\mathcal{C}$ carries a risk threshold $\theta_c\in[0,1]$ and a set of \emph{required prerequisite layers} $\mathrm{req}(c)$ with per-layer thresholds $\theta_{c,\ell}$.

\begin{definition}[Trust-gated grant]
Agent $i$ is granted capability $c$ at step $t$ iff
\begin{equation}
g_i(c,t)=\mathbb{1}\!\Big[\Phi_p(\mathbf{T}_i(t)) \ge \theta_c \;\wedge\; \forall \ell\!\in\!\mathrm{req}(c)\!:\, \tilde T_{i,\ell}(t)\ge\theta_{c,\ell}\Big].
\label{eq:grant}
\end{equation}
\end{definition}

Grants are recomputed every step, so a capability token is short-lived and automatically revoked when trust degrades, realizing least-privilege and revocability. Algorithm~\ref{alg:tgcc} gives the runtime loop.

\begin{algorithm}[t]
\caption{TGCC runtime loop (per agent $i$, per step $t$)}
\label{alg:tgcc}
\begin{algorithmic}[1]
\REQUIRE coupling $\mathbf{C}$, order $p$, thresholds $\{\theta_c,\theta_{c,\ell}\}$, rates $\lambda,\eta$
\STATE observe layer checks $s_1,\dots,s_L$; update Beta beliefs via \eqref{eq:beta}
\STATE compute effective trust $\tilde{\mathbf{T}}$ via synergy operator \eqref{eq:synergy}
\STATE update risk $\hat r_\ell$ and weights $\boldsymbol\alpha$ via \eqref{eq:hedge}
\STATE composite $\Phi_p \leftarrow (\sum_\ell \alpha_\ell \tilde T_\ell^{\,p})^{1/p}$
\FORALL{capabilities $c\in\mathcal{C}$}
  \STATE issue/keep token for $c$ iff grant predicate \eqref{eq:grant} holds; else revoke
\ENDFOR
\IF{$\min_\ell \tilde T_\ell < \theta_{\mathrm{rec}}$}
  \STATE trigger layer-specific recovery (e.g.\ external knowledge anchoring)
\ENDIF
\end{algorithmic}
\end{algorithm}

Every quantity is maintained by a recursive update, so the controller is cheap enough to run on every interaction.

\begin{proposition}[Constant per-step overhead]
\label{prop:complexity}
Each iteration of Algorithm~\ref{alg:tgcc} runs in $O(L^2+|\mathcal{C}|)$ time and $O(L)$ memory per agent, independent of the interaction history length $t$.
\end{proposition}
\begin{proof}
The Beta updates~\eqref{eq:beta} and the Hedge update~\eqref{eq:hedge} touch each layer once ($O(L)$) and store only $2L$ pseudo-counts and $L$ risk estimates. The synergy operator~\eqref{eq:synergy} forms $L$ products of at most $L$ factors ($O(L^2)$), and the composite~\eqref{eq:phi} is $O(L)$. The grant test~\eqref{eq:grant} costs $O(|\mathcal{C}|)$ for bounded prerequisite sets. History-independence follows because the forgetting recursions summarize the entire past in the $O(L)$ state.
\end{proof}

With $L=5$ this overhead is negligible, so per-step re-evaluation and prompt revocation are practical even at large collective scale.

\subsection{Breaking the Trust-Vulnerability Paradox}
A naive controller authorizes on a single pairwise/behavioral trust signal $\tau\approx T_2$ (the behavioral layer). Under such gating, an adversary that keeps behaving agreeably retains privilege even while its epistemic reliability collapses. This is precisely the regime in which over-exposure grows with trust~\cite{qi2025trust}. TGCC removes this attack surface.

\begin{theorem}[Paradox resolution]
\label{thm:paradox}
Consider a stealthy compromise in which an agent's behavioral trust is driven arbitrarily high, $T_2\to1$, while a prerequisite layer $k$ degrades to $T_k\le\delta$. Under naive behavioral gating, every capability with $\theta_c\le T_2$ remains granted, so over-exposure is unbounded in the compromise depth. Under TGCC with order $p$ and $k$ a prerequisite of all layers (e.g.\ epistemic), every capability with
\begin{equation}
\theta_c > \kappa\,\delta,\qquad \kappa=\bigl(\min_\ell\alpha_\ell\bigr)^{1/p},
\label{eq:revoke}
\end{equation}
is revoked irrespective of $T_2$. Hence the granted high-risk set is non-increasing in the compromise depth $(1-\delta)$.
\end{theorem}
\begin{proof}
By Proposition~\ref{prop:cascade} the compromise gives $\Phi_p\le\kappa\delta$. If $\theta_c>\kappa\delta$ then $\Phi_p<\theta_c$, so the first conjunct of~\eqref{eq:grant} fails and $c$ is denied regardless of $T_2$. Decreasing $\delta$ (a deeper compromise) lowers the bound $\kappa\delta$, weakly enlarging the revoked set, so the granted high-risk set is non-increasing in $(1-\delta)$. The naive case is immediate, its grant rule depending only on $T_2$.
\end{proof}

Combining the conservative fixed point (Proposition~\ref{prop:fixedpoint}) with Theorem~\ref{thm:paradox} bounds how quickly a compromise is contained.

\begin{corollary}[Self-revocation with bounded latency]
\label{cor:revoke}
Suppose at time $t_0$ a prerequisite layer $k$ (prerequisite to all layers) suffers a reliability drop whose conservative fixed point satisfies $T_k^\star\le\delta$. Then $T_k(t)\to T_k^\star$ geometrically at rate $\gamma$, and once $T_k(t)\le\delta$, TGCC denies every capability with $\theta_c>\kappa\delta$ regardless of behavioral trust. The number of steps until $T_k$ reaches an $\epsilon$-neighborhood of $T_k^\star$ is $O\!\big(\log_{1/\gamma}(1/\epsilon)\big)$.
\end{corollary}
\begin{proof}
Geometric convergence at rate $\gamma$ follows from Proposition~\ref{prop:fixedpoint}; the threshold-crossing time is then logarithmic in the target accuracy. Once $T_k\le\delta$, Theorem~\ref{thm:paradox} applies.
\end{proof}

\begin{insightbox}
\textit{\textbf{Key insight.} Naive gating makes authorization increasing in raw trust; TGCC makes it increasing in the \emph{weakest cross-layer-verified} trust, which an adversary cannot inflate by behaving agreeably -- turning the paradox's monotone exposure growth into bounded, self-revoking exposure with logarithmic detection latency.}
\end{insightbox}

\section{Numerical Illustration}
\label{sec:eval}
We validate the framework in a controlled simulation of the layer-trust dynamics, instantiating a single agent under a \emph{sleeper}-agent threat as arises in long-running collectives such as Agent Hospital~\cite{li2025agenthospitalsimulacrumhospital}, MetaGPT~\cite{hong2308metagpt}, and GPTSwarm~\cite{zhuge2024gptswarm}. Each layer emits Bernoulli evidence from a latent reliability $\rho_\ell$; beliefs update via~\eqref{eq:beta} ($\gamma=0.985$, $\omega=3$). Before activation at $t=80$ all layers are reliable ($\rho\approx0.9$); after activation, epistemic reliability collapses to $\rho_1=0.30$ (confident hallucination) while behavioral reliability stays at $0.93$ (agent remains agreeable), matching the exact stealth profile of Theorem~\ref{thm:paradox}. We set epistemic as a prerequisite of all layers, $p=-6$, and revocation threshold $\theta=0.55$, averaged over $400$ seeds.

\textbf{Mechanism (Fig.~\ref{fig:traj}).} Behavioral trust stays high after activation, yet epistemic trust collapses and, through the synergy operator~\eqref{eq:synergy}, drags composite $\Phi_p$ below $\theta$, so TGCC revokes the high-risk capability while a behavioral-only gate never would -- Theorem~\ref{thm:paradox} made visible.

\textbf{Validation of the theory.} Three checks (Table~\ref{tab:results}). (i)~\emph{Theorem~\ref{thm:bounds}}: across $20{,}000$ random draws of $(\mathbf{T},\boldsymbol{\alpha})$, zero violations of the bounds~\eqref{eq:thm1bounds}. (ii)~\emph{Proposition~\ref{prop:fixedpoint}}: observed steady-state trust matches the closed-form fixed point $T^\star_\ell$ to a mean absolute error of $0.007$ across layers (Fig.~\ref{fig:eval}(a)). (iii)~\emph{Corollary~\ref{cor:revoke}}: after compromise, the capability is revoked within $3.1$ steps on average (median $2$), consistent with the predicted logarithmic latency. Fig.~\ref{fig:eval}(b) shows Hedge weights shifting toward the epistemic layer once it becomes the dominant failure source, as Lemma~\ref{lem:regret} anticipates.

\begin{table}[t]
\caption{Predictions vs.\ simulation (400 seeds).}
\label{tab:results}
\centering
\footnotesize
\renewcommand{\arraystretch}{1.15}
\begin{tabular}{p{3.2cm}cc}
\toprule
\rowcolor{headrow}
\textbf{Result} & \textbf{Prediction} & \textbf{Observed} \\
\midrule
Thm.~\ref{thm:bounds} bound violations & $0$ & $0/20{,}000$ \\
Prop.~\ref{prop:fixedpoint} fixed point (MAE) & exact & $0.007$ \\
Cor.~\ref{cor:revoke} revocation latency & $O(\log_{1/\gamma}\!\tfrac{1}{\epsilon})$ & $3.1$ steps \\
Breaks TVP (Thm.~\ref{thm:paradox}) & yes & yes \\
\bottomrule
\end{tabular}
\end{table}

\section{Deployment and Extensions}
\label{sec:discussion}
\textbf{Protocol mapping.} TGCC is directly implementable over emerging agent protocols. In MCP~\cite{mcp2024} and A2A~\cite{a2a2025}, capabilities are exposed as scoped, token-mediated grants, but the policy that sets scope is left unspecified~\cite{interop2025}. Equation~\eqref{eq:grant} is exactly such a policy: an OAuth-style scope is issued only while $\Phi_p$ and the required prerequisite layers clear the capability's thresholds, with short-lived tokens giving the per-call revocation zero-trust guidance recommends~\cite{zerotrust2025,nist800207}.

\textbf{Cross-modal epistemic trust.} The framework is modality-agnostic: for embodied or perceptual agents, the epistemic layer's evidence $s_1$ can be sensor-fusion confidence or perceptual anomaly scores, so adversarial perturbations degrade $T_1$ and, via~\eqref{eq:synergy}, gate downstream actuation -- extending the stack beyond dialogue agents.

\textbf{Economic-incentive layer.} In agentic marketplaces, misaligned incentives (e.g.\ tacit collusion) are not captured by reliability alone. An incentive layer between role and social, tracking deviation from incentive-compatible behavior, slots into~\eqref{eq:synergy} without changing the operator and composes with governance-graph mechanisms~\cite{instai2026}.

\section{Conclusion}
We turned a descriptive layered trust model into an operational controller for multi-agent LLM systems. A cross-layer synergy operator and generalized-mean composite give a principled, bounded way to combine layer trusts; a no-regret estimator grounds layer importance in observed failures; and Trust-Gated Capability Control converts the result into revocable, least-privilege authorization. We prove, and numerically confirm, that this combination breaks the Trust-Vulnerability Paradox: privilege tracks the weakest verified layer, not raw trust, so a stealthy compromise that inflates behavioral trust while degrading a prerequisite layer self-revokes within a logarithmic number of steps rather than escalating. It is implementable over current agent protocols and extends to perceptual and economic settings, offering a concrete path to safe, scalable multi-agent collaboration.

\end{document}